\documentclass[12pt]{article}%
\usepackage{amsfonts, amsmath, amsthm, amssymb, bbm}
\usepackage{multicol,enumitem,graphicx,stmaryrd,color}
\usepackage{url,fancyhdr}
\usepackage[a4paper, top=2.5cm, bottom=2.5cm, left=2.2cm, right=2.2cm]{geometry}
\usepackage{algorithm, float}

\usepackage[driverfallback=hypertex,pagebackref=true,colorlinks]{hyperref}
	\hypersetup{linkcolor=[rgb]{.7,0,0}}
	\hypersetup{citecolor=[rgb]{0,.7,0}}
	\hypersetup{urlcolor=[rgb]{.7,0,.7}}

\newtheorem{theorem}{Theorem}
\newtheorem{lemma}[theorem]{Lemma}
\newtheorem{definition}[theorem]{Definition}
\newtheorem{observation}[theorem]{Observation}
\newtheorem{conjecture}[theorem]{Conjecture}

\newtheorem*{remark}{Remarks}
\newtheorem{proposition}[theorem]{Proposition}

\renewcommand{\implies}{\Longrightarrow}
\newcommand{\st}{\,|\,} 
\newcommand{\N}{\mathbb{N}} 

\newcommand{\set}[1]{\left\{ #1 \right\}}   
\newcommand{\brac}[1]{\left( #1 \right)}    
\newcommand{\sqbrac}[1]{\left[ #1 \right]}    

\newcommand{\val}{\text{val}} 	

\newcommand{\norm}[1]{\left\vert #1 \right\vert}

\newcommand{\mc}{\mathcal}

\newcommand{\F}{\mathbb F}

\newcommand{\NormF}{\norm{\mathbb F}}
\newcommand{\probname}{\text{Tensor\textsubscript{k}}}

\newcommand{\DTIME}{\mathrm{DTIME}}
\newcommand{\NTIME}{\mathrm{NTIME}}

\begin{document}

\title{Block Rigidity: Strong Multiplayer Parallel Repetition implies Super-Linear Lower Bounds for Turing Machines}
\author{Kunal Mittal\thanks{Department of Computer Science, Princeton University. Research supported by the Simons Collaboration on Algorithms and Geometry, by a Simons Investigator Award and by the National Science Foundation grants No. CCF-1714779, CCF-2007462.} \and Ran Raz\footnotemark[1]}
\date{}
\maketitle

\begin{abstract}
We prove that a sufficiently strong parallel repetition theorem for a special case of multiplayer (multiprover) games implies super-linear lower bounds for multi-tape Turing machines with advice.
To the best of our knowledge, this is 
the first connection between parallel repetition and lower bounds for time complexity
and the first major potential implication of a parallel repetition theorem with more than two players.

Along the way to proving this result, we define and initiate a study of {\it block rigidity}, a weakening of Valiant's notion of {\it rigidity}~\cite{Val77}. While rigidity was originally defined for matrices, or, equivalently, for (multi-output) linear functions, we extend and study both rigidity and block rigidity for general (multi-output) functions.
Using techniques of Paul, Pippenger, Szemer{\'{e}}di and Trotter \cite{PPST83},
we show that a block-rigid function cannot be computed by multi-tape Turing machines that run in linear (or slightly super-linear) time, even in the non-uniform setting, where the machine gets an arbitrary advice tape.

We then describe a class of multiplayer games, such that, a sufficiently strong parallel repetition theorem for that class of games implies an explicit block-rigid function.
The games in that class have the following property that may be of independent interest: for every random string for the verifier (which, in particular, determines the vector of queries to the players), there is a unique correct answer for each of the players, and the verifier accepts if and only if all answers are correct. We refer to such games as {\it independent games}.
The theorem that we need is that parallel repetition reduces the value of games in this class from $v$ to $v^{\Omega(n)}$, where $n$ is the number of repetitions.

As another application of block rigidity, we show conditional size-depth tradeoffs for boolean circuits, where the gates compute arbitrary functions over large sets.

\end{abstract}
\newpage
\section{Introduction}

We study relations between three seemingly unrelated topics: parallel repetition of multiplayer games, a variant of Valiant's notion of rigidity, that we refer to as block rigidity, and proving super-linear lower bounds for Turing machines with advice.

\subsection{Super-Linear Lower Bounds for Turing Machines}

Deterministic multi-tape Turing machines are the standard model of computation for defining time-complexity classes. Lower bounds for the running time of such machines are known by the time-hierarchy theorem~\cite{HS65}, using a diagonalization argument.
Moreover, the seminal work of Paul, Pippenger, Szemer{\'{e}}di and Trotter 
gives a separation of non-deterministic linear time from deterministic linear time, for multi-tape Turing machines~\cite{PPST83} (using ideas from \cite{HPV77} and \cite{PR80}). That is, it shows that
$\DTIME(n) \subsetneq \NTIME(n)$. 
This result has been slightly improved to give $\DTIME(n\sqrt{\log^*{n}}) \subsetneq \NTIME(n\sqrt{\log^*{n}})$~\cite{San01}.

However, the above mentioned lower bounds do not hold in the non-uniform setting, where the machines are allowed to use arbitrary advice depending on the length of the input. In the non-uniform setting, no super-linear lower bound is known for the running time of deterministic multi-tape Turing machines. Moreover, such bounds are not known even for a multi-output function.



\subsection{Block Rigidity}

The concept of matrix rigidity was introduced by Valiant as a means to prove super-linear lower bounds against circuits of logarithmic depth~\cite{Val77}. Since then, it has also found applications in communication complexity~\cite{Raz89} (see also~\cite{Wun12,Lok01}).
We extend Valiant's notion of rigid matrices to the concept of {\it rigid functions}, and further to {\it block-rigid functions}.
We believe that these notions are of independent interest.
We note that block-rigidity is a weaker condition than rigidity and hence it may be easier to find explicit block-rigid functions.
Further, our result gives a new application of rigidity.

Over a field $\F$, a matrix $A\in \F^{n\times n}$ is said to be an $(r,s)$-rigid matrix if it is not possible to reduce the rank of $A$ to at most $r$, by changing at most $s$ entries in each row of $A$. Valiant showed that if $A\in \F^{n\times n}$ is $(\epsilon n, n^{\epsilon})$-rigid for some constant $\epsilon >0$, then $A$ is not computable by a linear-circuit of logarithmic depth and linear size. As in many problems in complexity, the challenge is to find explicit rigid matrices. By explicit, we mean that a polynomial time deterministic Turing machine should be able to output a rigid matrix $A\in \F^{n\times n}$ on input~$1^n$. The best known bounds on explicit rigid matrices are far from what is needed to get super-linear circuit lower bounds 
(see \cite{Fri93,SS97,SSS97,Lok00,Lok06,KLPS14,AW17,GT18,AC19,BHPT20}).

We extend the above definition to functions $f:\set{0,1}^n\to \set{0,1}^n$, by saying that $f$ is \emph{not} an $(r,s)$-rigid function if there exists a subset $X\subseteq \set{0,1}^n$ of size at least $2^{n-r}$, such that over $X$, each output bit of $f$ can be written as a function of some $s$ input bits (see Definition~\ref{def:rig_func}).
By a simple counting argument (see Proposition~\ref{prop:rig_func_exist}), it follows that random functions are rigid with good probability.

We further extend this definition to what we call block-rigid functions (see Definition~\ref{def:block_rig_func}). For this, we'll consider vectors $x\in \set{0,1}^{nk}$, which are thought of as composed of $k$ blocks, each of size $n$.
We say that a function $f:\set{0,1}^{nk}\to \set{0,1}^{nk}$ is \emph{not} an $(r,s)$-block-rigid function, if there exists a subset $X\subseteq \set{0,1}^{nk}$ of size at least $2^{nk-r}$, such that over $X$, each \emph{output block} of $f$ can be written as a function of some $s$ \emph{input blocks}.

We conjecture that it is possible to obtain large block-rigid functions, using smaller rigid functions. For a function $f:\set{0,1}^k\to \set{0,1}^k$, we define the function $f^{\otimes n}:\set{0,1}^{nk}\to \set{0,1}^{nk}$ as follows.
For each $x = (x_{ij})_{i\in[n],j\in [k]} \in \set{0,1}^{nk}$, we define $f^{\otimes n}(x)$ to be the vector obtained by applying $f$ to $(x_{i1},\dots,x_{ik})$, in place for each $i\in[n]$ (see Definition~\ref{def:f_tens_idn}).

\begin{conjecture}\label{intro_conj:rig_amp_to_block_rig}
	There exists a universal constant $c>0$ such that the following is true.
	Let $f: \set{0,1}^k\to \set{0,1}^k$ be an $(r, s)$-rigid function, and $n\in \N$. Then, $f^{\otimes n}: \set{0,1}^{nk}\to \set{0,1}^{nk}$ is a $(cnr, cs)$-block-rigid function.
\end{conjecture}

We prove the following theorem. It is restated and proved as Theorem \ref{thm:tm_lb_explicit} in Section \ref{sec:tms}.
\begin{theorem}\label{intro_thm:tm_lb_explicit}
Let $t:\N\to\N$ be any function such that $t(n) = \omega(n)$.
Assuming Conjecture $\ref{intro_conj:rig_amp_to_block_rig}$, there exists an (explicitly given) function $f:\set{0,1}^*\to \set{0,1}^*$ such that
\begin{enumerate}
	\item On inputs $x$ of length $n$ bits, the output $f(x)$ is of length at most $n$ bits.
	\item The function $f$ is computable by a multi-tape deterministic Turing machine that runs in time $O(t(n))$ on inputs of length $n$.
	\item The function $f$ is not computable by any multi-tape deterministic Turing machine that takes advice and runs in time $O(n)$ on inputs of length $n$.
\end{enumerate}
\end{theorem}

More generally, we show that families of block-rigid functions cannot be computed by non-uniform Turing machines running in linear-time.
This makes it interesting to find such families that are computable in polynomial time. 
The following theorem is restated as Theorem \ref{thm:tm_lb_blk_rig}.
\begin{theorem}\label{intro_thm:tm_lb_blk_rig}
	Let $k: \N \to \N$ be a function such that $k(n)=\omega(1)$ and $k(n)=2^{o(n)}$, and $f_n:\set{0,1}^{nk(n)}\to\set{0,1}^{nk(n)}$ be a family of $(\epsilon nk(n), \epsilon k(n))$-block-rigid functions, for some constant $\epsilon >0$.
	Let $M$ be any multi-tape deterministic linear-time Turing machine that takes advice.
	Then, there exists $n\in \N$, and $x\in \set{0,1}^{nk(n)}$, such that $M(x)\not= f_n(x)$.
\end{theorem}

As another application, based on Conjecture \ref{intro_conj:rig_amp_to_block_rig}, we show size-depth tradeoffs for boolean circuits, where the gates compute arbitrary functions over large (with respect to the input size) sets (see Section \ref{sec:ckt_lb}).


\subsection{Parallel Repetition}
In a $k$-player game $\mc G$, questions $(x_1,\dots,x_k)$ are chosen from some joint distribution $\mu$.
For each $j\in [k]$, player $j$ is given $x_j$ and gives an answer $a_j$ that depends only on $x_j$.
The players are said to win if their answers satisfy a fixed predicate $V(x_1,\dots,x_k,a_1,\dots,a_k)$.
We note that $V$ might be randomized, that is, it might depend on some random string that is sampled independently of $(x_1,\dots,x_k)$.
The value of the game $\val(\mc G)$ is defined to be the maximum winning probability over the possible strategies of the players.

It is natural to consider the parallel repetition $\mc G^{\otimes n}$ of such a game $\mc G$.
Now, the questions $(x_1^{(i)}, \dots, x_k^{(i)})$ are chosen from $\mu$, independently for each $i\in[n]$.
For each $j\in [k]$, player $j$ is given $(x_j^{(1)}, \dots, x_j^{(n)})$ and gives answers $(a_j^{(1)}, \dots, a_j^{(n)})$.
The players are said to win if the answers satisfy $V(x_1^{(i)},\dots,x_k^{(i)},a_1^{(i)},\dots,a_k^{(i)})$ for every $i\in [n]$.
The value of the game $\val(\mc G^{\otimes n})$ is defined to be the maximum winning probability over the possible strategies of the players.
Note that the players are allowed to correlate their answers to different repetitions of the game.

Parallel repetition of games was first studied 
in~\cite{FRS94}, owing to its relation with multiprover interactive proofs~\cite{BOGKW88}.
It was hoped that the value $\val(\mc G^{\otimes n})$ of the repeated game goes down as $\val(\mc G)^n$.
However, this is not the case, as shown in \cite{For89,Fei91,FV02,Raz11}.

A lot is known about parallel repetition of 2-player games.
The, so called, parallel repetition theorem, first proved by Raz~\cite{Raz98} and further
simplified and improved by Holenstein~\cite{Hol09}, shows that if $\val(\mc G) < 1$, then $\val(\mc G^{\otimes n}) \leq 2^{-\Omega(n/s)}$, where $s$ is the length of the answers given by the players.
The bounds in this theorem were later made tight even for the case when the initial game has small value (see~\cite{DS14} and~\cite{BG15}).

Much less is known for $k$-player games with $k\geq 3$.
Verbitsky \cite{Ver96} showed that if $\val(G)<1$, then the value of the the repeated game goes down to zero as $n$ grows larger.
The result shows a very weak rate of decay, approximately equal to $\frac{1}{\alpha(n)}$, where $\alpha$ is the inverse-Ackermann function, owing to the use of the density Hales-Jewett theorem (see \cite{FK91} and \cite{Pol12}).
A recent result by Dinur, Harsha, Venkat and Yuen \cite{DHVY17} shows exponential decay, but only in the special case of what they call \emph{expanding games}.
This approach fails when the input to the players have strong correlations.

In this paper (see Section \ref{sec:par_rep}), we show that a sufficiently strong parallel repetition theorem for multiplayer games implies Conjecture~\ref{intro_conj:rig_amp_to_block_rig}.
The following theorem is proved formally as Theorem \ref{thm:parrep_implies_rigidity} in Section \ref{sec:par_rep}.

\begin{theorem}
	There exists a family $\set{\mc G_{\mc S,k}}$ of $k$-player games (where $\mc S$ is some parameter), such that a strong parallel repetition theorem for all games in  $\set{\mc G_{\mc S,k}}$ implies Conjecture~\ref{intro_conj:rig_amp_to_block_rig}.
\end{theorem}

Although the games in this family do not fit into the framework of \cite{DHVY17}, they satisfy some very special properties. Every $k$-player game in the family satisfies the following:
\begin{enumerate}
	\item The questions to the $k$-players are chosen as follows: First, $k$ bits, $x_1,\dots,x_k \in \set{0,1}$, are drawn uniformly and independently. Each of the $k$-players sees some subset of these $k$-bits.
	\item The predicate $V$ satisfies the condition that on fixing the bits $x_1,\dots,x_k$, there is a unique accepting answer for each player (independently of all other answers) and the verifier accepts if every player answers with the accepting answer. We refer to games that satisfy this property as {\it independent games}.
\end{enumerate}

We believe that these properties may allow us to prove strong upper bounds on the value of parallel repetition of such games, despite our lack of understanding of multiplayer parallel repetition. 
The bounds that we need are that parallel repetition reduces the value of such games  from $v$ to $v^{\Omega(n)}$, where $n$ is the number of repetitions (as is proved in~\cite{DS14} and~\cite{BG15} for 2-player games).

\subsection{Open Problems}
\begin{enumerate}
	\item The main open problem is to make progress towards proving Conjecture \ref{intro_conj:rig_amp_to_block_rig}, possibly using the framework of parallel repetition. The remarks after Theorem \ref{thm:tm_lb_explicit} mention some weaker statements that suffice for our applications. The examples of matrix-transpose and matrix-product in Section \ref{sec:matrix_problems} also serve as interesting problems.
	\item Our techniques, which are based on \cite{PPST83}, heavily exploit the fact that the Turing machines have one-dimensional tapes. Time-space lower bounds for satisfiability in the case of multi-tape Turing machines with random access \cite{FLvMV05}, and Turing machines with one $d$-dimensional tape \cite{vMR05}, are known. Extending such results to the non-uniform setting is an interesting open problem.
	\item The question of whether a rigid-matrix $A\in \F_2^{n\times n}$ is rigid when seen as a function $A:\set{0,1}^n\to\set{0,1}^n$ is very interesting (see Section \ref{sec:rig_mat_vs_func}). This question is closely related to  a Conjecture of Jukna and Schnitger \cite{JS11} on the linearization of depth-2 circuits. This is also related to the question of whether data structures for linear problems can be optimally linearized (see \cite{DGW19}). We note that there are known examples of linear problems for which the best known data-structures are non-linear, without any known linear data-structure achieving the same bounds (see \cite{KU11}).
\end{enumerate}

\section{Preliminaries}
Let $\N=\set{1,2,3,\dots}$ be the set of all natural numbers. For any $k\in \N$, we use $[k]$ to denote the set $\set{1,2,\dots,k}$.

We use $\F$ to denote an arbitrary finite field, and $\F_2$ to denote the finite field on two elements.

Let $x\in \set{0,1}^k$.
For $i\in [k]$, we use $x_i$ to denote the $i$\textsuperscript{th} coordinate of $x$.
For $S\subseteq [k]$, we denote by $x|_S$ the vector $(x_i)_{i\in S}$, which is the restriction of $x$ to coordinates in $S$.

We also consider vectors $x\in \set{0,1}^{nk}$, for some $n\in \N$.
We think of these as composed of $k$ blocks, each consisting of a vector in $\set{0,1}^n$.
That is, $x = (x_{ij})_{i\in [n], j\in [k]}$.
By abuse of notation, for $S\subseteq [k]$, we denote by $x|_S$ the vector $(x_{ij})_{i\in [n], j\in S}$, which is the restriction of $x$ to the blocks indexed by $S$.

Let $A\in \F^{nk\times nk}$ be an $nk\times nk$ matrix.
We think of $A$ as a block-matrix consisting of $k^2$ blocks, each block being an $n\times n$ matrix.
That is, $A = (A_{ij})_{i, j\in [k]}$, where for all $i, j\in [k]$, $A_{ij}\in \F^{n\times n}$.
For each $i\in [k]$, we call $(A_{ij})_{j\in [k]}$ the $i$\textsuperscript{th} block-row of $A$.

For every $n\in \N$, we define $\log^* n = \min\{\ell\in \N\cup\set{0} : \underbrace{\log_2 \log_2 \dots \log_2}_{\ell \text{ times}}{n} \leq 1 \}.$

\section{Rigidity and Block Rigidity} \label{sec:rig_and_block_rig}

\subsection{Rigidity}\label{subsec:rig}

The concept of matrix rigidity was introduced by Valiant \cite{Val77}.
It is defined as follows.

\begin{definition}\label{def:rig_mat}
	A matrix $A \in \F^{n\times n}$ is said to be an $(r, s)$-rigid matrix if it cannot be written as $A = B+C$, where $B$ has rank at most $r$, and $C$ has at most $s$ non-zero entries in each row.
\end{definition}

Valiant \cite{Val77} showed the existence of rigid matrices by a simple counting argument.
For the sake of completeness, we include this proof.

\begin{proposition}\label{prop:rig_mat_exist}
	For any constant $0<\epsilon\leq\frac{1}{8}$, and any $n\in \N$, there exists a matrix $A\in \F^{n\times n}$ that is an $(\epsilon n, \epsilon n)$-rigid matrix.
\end{proposition}
\begin{proof}
	Fix any $0<\epsilon\leq\frac{1}{8}$. We bound the number of $n\times n$ matrices that are not $(\epsilon n, \epsilon n)$-rigid matrices.
	\begin{enumerate}
		\item Any $n\times n$ matrix with rank at most $r$ can be written as the product of an $n\times r$ and an $r \times n$ matrix. Hence, the number of matrices of rank at most $\epsilon n$ is at most $\NormF^{2\epsilon n^2} \leq \NormF^{\frac{n^2}{4}}$.
		\item The number of matrices that have at most $\epsilon n$ non-zero entries in each row is at most \[\brac{\binom{n}{\epsilon n}\NormF^{\epsilon n}}^n \leq \brac{\frac{e}{\epsilon}}^{\epsilon n^2}\NormF^{\epsilon n^2} = \NormF^{\epsilon n^2(1 + \log_{\NormF}{ \frac{e}{\epsilon}} )} < \NormF^{\frac{3n^2}{4}}.\]
		We used the binomial estimate $\binom{n}{r} \leq \brac{\frac{en}{r}}^r$.
	\end{enumerate}
	Since each matrix that is not an $(r, s)$-rigid matrix can be written as the sum of a matrix with rank at most $r$, and a matrix with at most $s$ non-zero entries in each row, the number of matrices that are not $(\epsilon n,\epsilon n)$-rigid matrices is strictly less than $\NormF^{\frac{n^2}{4}} \cdot \NormF^{\frac{3n^2}{4}} = \NormF^{n^2}$, which is the total number of $n\times n$ matrices.
\end{proof}

Observe that a matrix $A \in \F_2^{n\times n}$ is not an $(r, s)$-rigid matrix if and only if there is a subspace $X\subseteq \F_2^n$ of dimension at least $n-r$, and a matrix $C$ with at most $s$ non-zero entries in each row, such that $Ax=Cx$ for all $x\in X$.

We use this formulation to extend the concept of rigidity to general functions.

\begin{definition}\label{def:rig_func}
	A function $f:\set{0,1}^n\to \set{0,1}^n$ is said to be an $(r, s)$-rigid function if for every subset $X\subseteq \set{0,1}^n$ of size at least $2^{n-r}$, and subsets $S_1,\dots,S_n\subseteq[n]$ of size $s$, and functions $g_1,\dots,g_n:\set{0,1}^s\to \set{0,1}$, there exists $x\in X$ such that $f(x) \not= (g_1(x|_{S_1}), g_2(x|_{S_2}), \dots, g_n(x|_{S_n}))$.
\end{definition}

Using a similar counting argument as in Proposition \ref{prop:rig_mat_exist}, we show the existence of rigid functions.

\begin{proposition}\label{prop:rig_func_exist}
	For any constant $0<\epsilon\leq\frac{1}{8}$, and any (large enough) integer $n$, there exists a function $f:\set{0,1}^n \to \set{0,1}^n$ that is an $(\epsilon n, \epsilon n)$-rigid function.
\end{proposition}
\begin{proof}
	Fix any constant $0<\epsilon\leq\frac{1}{8}$, and any integer $n \geq \frac{2}{\epsilon}$.
	We count the number of functions $f:\set{0,1}^n\to \set{0,1}^n$ that are not $(\epsilon n, \epsilon n)$-rigid functions.
	\begin{enumerate}
		\item Note that in Definition \ref{def:rig_func}, it is without loss of generality to assume that $\norm{X} = 2^{n-\epsilon n}$. The number subsets $X\subseteq \set{0,1}^n$ of size $2^{n-\epsilon n}$ is $\binom{2^n}{2^{n-\epsilon n}} \leq \brac{e2^{\epsilon n}}^{2^{n-\epsilon n}} < 2^{2\epsilon n 2^{n-\epsilon n}} \leq 2^{\frac{n}{4} 2^{n-\epsilon n}}.$
		\item The number of subsets $S_1,\dots,S_n \subseteq[n]$ of size $\epsilon n$ is at most $ \binom{n}{\epsilon n}^n < n^{\epsilon n^2} < 2^{\frac{n}{4} 2^{n-\epsilon n}}$.
		\item The number of functions $g_1,\dots,g_n: \set{0,1}^{\epsilon n}\to \set{0,1}$ is at most $ 2^{n2^{\epsilon n}} < 2^{\frac{n}{4} 2^{n-\epsilon n}}$.
		\item The number of choices for values of $f$ on $\set{0,1}^n \setminus X$ is at most $2^{n(2^n-\norm{X})} \leq 2^{n(2^n-2^{n-\epsilon n})}$.
	\end{enumerate}
Hence, the total number of functions $f: \set{0,1}^n\to 	\set{0,1}^n$ that are not $(\epsilon n, \epsilon n)$-rigid functions is strictly less than
$ \brac{2^{\frac{n}{4} 2^{n-\epsilon n}}}^3 2^{n2^n-n2^{n-\epsilon n}} < 2^{n2^n}.$
\end{proof}

\subsection{Block Rigidity}\label{subsec:block_rig}

In this section, we introduce the notion of block rigidity.

\begin{definition}\label{def:block_rig_mat}
	A matrix $A \in  \F^{nk\times nk}$ is said to be an $(r,s)$-block-rigid matrix if it cannot be written as $A = B+C$, where $B$ has rank at most $r$, and $C$ has at most $s$ non-zero matrices in each block-row.
\end{definition}

Observe that if $A\in \F^{nk\times nk}$ is an $(r,ns)$-rigid matrix, then it is also $(r,s)$-block-rigid matrix.
Combining this with Proposition \ref{prop:rig_mat_exist}, we get the following.

\begin{observation}
	For any constant $0<\epsilon\leq\frac{1}{8}$, and positive integers $n,k$, there exists an $(\epsilon nk, \epsilon k)$-block-rigid matrix $A\in  \F^{nk\times nk}$.
\end{observation}

Following the definition of rigid-functions in Section \ref{subsec:rig}, we define block-rigid functions as follows.

\begin{definition}\label{def:block_rig_func}
	A function $f: \set{0,1}^{nk}\to \set{0,1}^{nk}$ is said to be an $(r,s)$-block-rigid function if for every subset $X\subseteq \set{0,1}^{nk}$ of size at least $2^{nk-r}$, and subsets $S_1,\dots,S_k\subseteq[k]$ of size $s$, and functions $g_1,\dots,g_k: \set{0,1}^{ns}\to  \set{0,1}^n$, there exists $x\in X$ such that $f(x) \not= (g_1(x|_{S_1}), g_2(x|_{S_2}), \dots, g_k(x|_{S_k}))$.
\end{definition}

Observe that if $f : \set{0,1}^{nk}\to \set{0,1}^{nk}$ is an $(r,ns)$-rigid function, then it is also $(r,s)$-block-rigid function.
Combining this with Proposition \ref{prop:rig_func_exist}, we get the following.

\begin{observation}
	For any constant $0<\epsilon\leq\frac{1}{8}$, and (large enough) integers $n,k$, there exists an $(\epsilon nk, \epsilon k)$-block-rigid function $f: \set{0,1}^{nk}\to \set{0,1}^{nk}$.
\end{observation}

Note that $n=1$ in the definition of block-rigid matrices (functions) gives the usual definition of rigid matrices (functions).
For our applications, we will mostly be interested in the case when $n$ is much larger than $k$.

\subsection{Rigidity Amplification}

A natural question is whether there is a way to amplify rigidity.
That is, given a rigid matrix (function), is there a way to obtain a larger matrix (function) which is rigid, or even block-rigid.

\begin{definition}\label{def:f_tens_idn}
	Let $f: \set{0,1}^k\to  \set{0,1}^k$ be any function. Define $f^{\otimes n}: \set{0,1}^{nk}\to \set{0,1}^{nk}$ as following. Let $x = (x_{ij})_{i\in[n],j\in[k]} \in  \set{0,1}^{nk}$ and $i\in [n], j\in [k]$. The $(i,j)$\textsuperscript{th} coordinate of $f^{\otimes n}(x)$ is defined to be the $j$\textsuperscript{th} coordinate of $f(x_{i1}, x_{i2},\dots, x_{ik})$.
\end{definition}

Basically, applying $f^{\otimes n}$ on $x\in  \set{0,1}^{nk}$ is the same as applying $f$ on $(x_{i1}, x_{i2},\dots, x_{ik})$, in place for each $i\in [n]$.
For a linear function given by matrix $A\in  \F_2^{k\times k}$, this operation corresponds to $A\otimes I_n$, where $I_n$ is the $n\times n$ identity matrix, and $\otimes$ denotes the Kronecker product of matrices.

It is easy to see that if $f$ is not rigid, then $f^{\otimes n}$ is not block-rigid.
\begin{observation}\label{obs:rig_amp_converse}
	Suppose $f: \set{0,1}^k\to \set{0,1}^k$ is not an $(r,s)$-rigid function. Then $f^{\otimes n}$ is not an $(nr,s)$-block-rigid function.
\end{observation}

The converse of Observation $\ref{obs:rig_amp_converse}$ is more interesting. We believe that it is true, and restate Conjecture \ref{intro_conj:rig_amp_to_block_rig} below.

\begin{conjecture}\label{conj:rig_amp_to_block_rig}
	There exists a universal constant $c>0$ such that the following is true.
	Let $f: \set{0,1}^k\to \set{0,1}^k$ be an $(r, s)$-rigid function, and $n\in \N$. Then, $f^{\otimes n}: \set{0,1}^{nk}\to \set{0,1}^{nk}$ is a $(cnr, cs)$-block-rigid function.
\end{conjecture}

\section{Parallel Repetition}\label{sec:par_rep}

In this section, we show an approach to prove Conjecture \ref{conj:rig_amp_to_block_rig} regarding rigidity amplification.
This is based on proving a strong parallel repetition theorem for a $k$-player game.

Fix some $k\in \N$, a function $f:  \set{0,1}^{k} \to  \set{0,1}^{k}$, an integer $1\leq s<k$, and $\mc S = (S_1,\dots,S_k)$, where each $S_i \subseteq[k]$ is of size $s$.
We define a $k$-player game $\mc G_{\mc S}$ as follows:

The $k$-players choose functions $g_1,\dots,g_k: \set{0,1}^s\to  \set{0,1}$, which we call a strategy.
A verifier chooses $x_1,\dots,x_k \in  \set{0,1}$ uniformly and independently.
Let $x = (x_1,\dots,x_k)\in \set{0,1}^k$.
For each $j\in[k]$, Player $j$ is given the input $x|_{S_j}$, and they answer $a_j =  g_j(x|_{S_j}) \in  \set{0,1}$.
The verifier accepts if and only if $f(x) = (a_1,\dots,a_k)$.
The goal of the players is to maximize the winning probability.
Formally, the value of the game is defined as \[\val(\mc G_{\mc S})  := \max_{g_1,\dots,g_k} \Pr_{x_1,\dots,x_k\in  \set{0,1}}\sqbrac{f(x) = \brac{g_1(x|_{S_1}),\dots, g_k(x|_{S_k}) }}.\]

The $n$-fold repetition of $\mc G_{\mc S}$, denoted by $\mc G_{\mc S}^{\otimes n}$ is defined as follows.
The players choose a strategy $g_1,\dots,g_k: \set{0,1}^{ns}\to  \set{0,1}^n$.
The verifier chooses $x_1,\dots,x_k \in  \set{0,1}^n$ uniformly and independently.
Let $x = (x_1,\dots,x_k)\in  \set{0,1}^{nk}$.
Player $j$ is given the input $x|_{S_j}$, and they answer $a_j=g_j(x|_{S_j}) \in  \set{0,1}^n$.
The verifier accepts if and only if $f^{\otimes n}(x) = (a_1,\dots,a_k)$.
That is, for each $i\in[n]$, $j\in[k]$, the $j$\textsuperscript{th} bit of $f(x_{i1},\dots,x_{ik})$ equals the $i$\textsuperscript{th} bit of $a_j$.
The value of this repeated game is
\[\val(\mc G_{\mc S}^{\otimes n})  := \max_{g_1,\dots,g_k} \Pr_{x_1,\dots,x_k\in  \set{0,1}^n}\sqbrac{f^{\otimes n}(x) = \brac{g_1(x|_{S_1}),\dots, g_k(x|_{S_k}) }}.\]

From Definition \ref{def:block_rig_func}, we get the following:
\begin{observation}\label{obs:block_rig_func_game_val}
	Let $f: \set{0,1}^{k} \to  \set{0,1}^{k}$ be a function, and $n\in \N$.
	Then, $f^{\otimes n}$ is an $(r,s)$-block-rigid function if and only if for every $\mc S = (S_1,\dots,S_k)$ with set sizes as $s$, $\val(\mc G_{\mc S}^{\otimes n}) < 2^{-r}$.
\end{observation}
\begin{proof}
	Let $f^{\otimes n}$ be an $(r,s)$-block-rigid function.
	Suppose, for the sake of contradiction, that $\mc S = (S_1,\dots,S_k)$ is such that $\val (\mc G_{\mc S}^{\otimes n}) \geq 2^{-r}$.
	Let the functions $g_1,\dots,g_k:  \set{0,1}^{ns}\to  \set{0,1}^n$ be an optimal strategy for the players.
	Define $X:= \set{x\in  \set{0,1}^{nk} \st f^{\otimes n}(x) = \brac{g_1(x|_{S_1}),\dots, g_k(x|_{S_k}) } }$.
	Then, $\norm{X} = \val(\mc G_{\mc S}) \cdot 2^{nk} \geq 2^{nk-r}$, which contradicts the block rigidity of $f^{\otimes n}$.

Conversely, suppose that $f^{\otimes n}$ is not $(r,s)$-block-rigid.
Then, there exists $X\subseteq  \set{0,1}^{nk}$ with $\norm{X}\geq 2^{nk-r}$, subsets $S_1,\dots,S_k\subseteq [k]$ of size $s$, and functions $g_1,\dots,g_k: \set{0,1}^{ns}\to \set{0,1}^n$, such that for all $x\in X$, $f^{\otimes n}(x) = (g(x|_{S_1}),\dots,g(x|_{S_k}))$.
Let $\mc S = (S_1,\dots,S_k)$, and suppose the players use strategy $g_1,\dots,g_k$.
Then, $\val(\mc G_{\mc S}^{\otimes n}) \geq \norm{X} \cdot 2^{-nk} \geq 2^{-r}$.
\end{proof}

In particular, for $n=1$, Observation \ref{obs:block_rig_func_game_val} gives the following:
\begin{observation}\label{obs:rig_func_game_val}
	A function  $f: \set{0,1}^{k} \to  \set{0,1}^{k}$ is an $(r,s)$-rigid-function if and only if for every $\mc S = (S_1,\dots,S_k)$ with set sizes as $s$, $\val(\mc G_{\mc S}) < 2^{-r}$.
\end{observation}

We conjecture the following strong parallel repetition theorem.
\begin{conjecture}\label{conj:par_rep}
	There exists a constant $c>0$ such that the following is true. Let $f: \set{0,1}^k\to \set{0,1}^k$ be any function, and $\mc S = (S_1,\dots,S_k)$ be such that for each $i\in [k]$, $S_i\subseteq[k]$ is of size $s$.
	Then, for all $n\in \N$, $\val(\mc G_{\mc S}^{\otimes n}) \leq (\val(\mc G_{\mc S}))^{cn}$.
\end{conjecture}

Combining Observation \ref{obs:block_rig_func_game_val} and \ref{obs:rig_func_game_val}, we get the following:
\begin{theorem} \label{thm:parrep_implies_rigidity}
	Conjecture \ref{conj:par_rep} $\implies$ Conjecture \ref{conj:rig_amp_to_block_rig}.
\end{theorem}

\begin{remark}
\begin{enumerate}[label=(\roman*)]
	\item By looking only at some particular player, it can be shown that if ${\val(\mc G_{\mc S})<1}$, then  $\val(\mc G_{\mc S}^{\otimes n}) \leq 2^{-\Omega(n)}$. In fact, such a result holds for all \emph{independent games}. The harder part seems to be showing strong parallel repetition when the initial game has small value.
	\item Observe that the game $\mc G_{\mc S}$ has a randomized predicate in the case $\cup_{j=1}^k S_j \not= [k]$. This condition can be removed (even for general independent games) by introducing a new player. This player is given the random string used by the verifier, and is always required to answer a single bit equal to zero. This maintains the independent game property, and ensures that the predicate used by the verifier is a deterministic function of the vector of input queries to the players. 
\end{enumerate}	
\end{remark}

\section{Turing Machine Lower Bounds}\label{sec:tms}

In this section, we show a conditional super-linear lower bound for multi-tape deterministic Turing machines that can take advice.

Without loss of generality, we only consider machines that have a separate read-only input tape.
We assume that the advice string, which is a function of the input length, is written on a separate advice tape at the beginning of computation.
We are interested in machines that compute multi-output functions.
For this, we assume that at the end of computation, the machine writes the entire output on a separate write-only output tape, and then halts.

We consider the following problem.
\begin{definition} Let $k:\N\to\N$ be a function. We define the problem $\probname$ as follows:
\begin{itemize}
	\item[-] \emph{Input:} $(f,x)$, where $f:\set{0,1}^k\to\set{0,1}^k$ is a function, and $x\in \set{0,1}^{nk}$, for some $n\in\N$ and $k = k(n)$.
	\item[-] \emph{Output:} $f^{\otimes n}(x) \in \set{0,1}^{nk}$.
\end{itemize}
The function $f$ is given as input in the form of its entire truth table. The input $x = (x_{ij})_{i\in [n],j\in[k]}$ is given in the order $(x_{11},\dots,x_{n1},x_{12},\dots,x_{n2},\dots,x_{1k},\dots,x_{nk})$. The total length of the input is $m(n) := 2^kk + nk$.
\end{definition}

We observe that if the function $k:\N\to \N$ grows very slowly with $n$, the problem $\probname$ can be solved by a deterministic Turing machine in slightly super-linear time.

\begin{observation}\label{obs:sup_lin_tm_ub}
	Let $k:\N\to\N$ be a function. There exists a deterministic Turing machine that solves the problem $\probname$ in time $O(nk2^k)$, on input $(f,x)$ of length $nk+2^kk$, where $k = k(n)$.
\end{observation}
\begin{proof}
	We note that applying $f^{\otimes n}$ on $x = (x_{ij})_{i\in[n],j\in[k]}$ is the same as applying $f$ on $(x_{i1}, x_{i2},\dots, x_{ik})$, in place for each $i\in [n]$. A Turing machine can do the following:
	\begin{enumerate}
		\item Find $n$ and $k$, using the fact that the description of $f$ is of length $2^kk$, and that of $x$ is of length $nk$.
		\item Rearrange the input so that for each $i\in [n]$, the part $(x_{i1}, x_{i2},\dots, x_{ik})$ is written consecutively on the tape.
		\item Using the truth table of $f$, compute the output for each such part.
		\item Rearrange the entire output back to the desired form.
	\end{enumerate}
	The total time taken is $O(nk+nk^2 + nk2^k+nk^2) = O(nk2^k)$.
\end{proof}

We now state and prove the main technical theorem of this section.

\begin{theorem}\label{thm:tm_advice_lb}
	Let $k:\N\to\N$ be a function such that $k(n) = \omega(1)$ and $k(n) = o(\log_2n)$.
	Let $m = m(n):= 2^kk + nk$, where $k = k(n)$.
	
	Suppose $M$ is a deterministic multi-tape Turing machine that takes advice, and runs in linear time in the length of its input.
	Assuming Conjecture \ref{conj:rig_amp_to_block_rig}, the machine $M$ does not solve the problem $\probname$ correctly for all inputs.
	That is, there exists $n\in \N$, and  $y = (f,x) \in \set{0,1}^m$ such that $M(y) \not= f^{\otimes n}(x)$.
\end{theorem}

There are two main technical ideas that will be useful to us.
The first is the notion of block-respecting Turing machines, defined by Hopcroft, Paul and Valiant \cite{HPV77}.
The second is a graph theoretic result, which was proven by Paul, Pippenger, Szemer{\'{e}}di and Trotter  \cite{PPST83}, and was used to show a separation between deterministic and non-deterministic linear time.

\begin{definition}\label{def:block_resp}
	Let $M$ be a Turing machine, and let $b:\N\to\N$ be a function.
	Partition the computation of $M$, on any input $y$ of length $m$, into time \emph{segments} of length $b(m)$, with the last segment having length at most $b(m)$.
	Also, partition each of the tapes of $M$ into \emph{blocks}, each consisting of $b(m)$ contiguous cells.
	
	We say that $M$ is block-respecting with respect to block size $b$, if on inputs of length $m$, the tape heads of $M$ cross blocks only at times that are integer multiples of $b(m)$.
\end{definition}

\begin{lemma}\label{lemma:hpv_block_resp}\cite{HPV77}
	Let $t:\N\to\N$ be a function, and $M$ be a multi-tape deterministic Turing machine running in time $t(m)$ on inputs of length $m$.
	Let $b:\N\to\N$ be a function such that $b(m)$ is computable from $1^m$ by a multi-tape deterministic Turing machine running in time $O(t(m))$.
	Then, the language recognized by $M$ is also recognized by a multi-tape deterministic Turing Machine $M'$, which runs in time $O(t(m))$, and is block-respecting with respect to $b$.
\end{lemma}

The rest of this section is devoted to the proof of Theorem \ref{thm:tm_advice_lb}.
Let $k:\N\to\N$ be a function such that $k(n) = \omega(1)$ and $k(n) = o(\log_2n)$.
Let $m = m(n) := 2^kk+nk$, where $k = k(n)$.

Suppose that $M$ is a deterministic multi-tape Turing Machine, which on input $y = (f,x) \in \set{0,1}^m$, takes advice, runs in time $O(m)$, and outputs $f^{\otimes n}(x)$.

Let $b:\N\to\N$ be a function such that $b(m) = n$.
By our assumption that $k(n) = o(\log_2 n)$, we can assume that the input $y = (f,x) \in \set{0,1}^m$ consists of $k+1$ blocks, where the first block contains $f$ (possibly padded by blank symbols to the left), and the remaining $k$ blocks contain $x$.
By Lemma \ref{lemma:hpv_block_resp}, we can further assume that $M$ is block-respecting with respect to $b$.
Note that we can assume $n$ to be a part of the advice, and hence we don't need to care about the computability of $b$.

For an input $y = (f,x) \in \set{0,1}^m$, the number of time segments for which $M$ runs on $x$ is at most $\frac{O(m)}{b(m)} = \frac{O(nk)}{n} = O(k) := a(n)$.

We define the computation graph $G_M(y)$, for input $y\in\set{0,1}^m$, as follows.

\begin{definition}\label{def:comp_graph}
The vertex set of $G_M(y)$ is defined to be $V_M(y) = \set{1,\dots,a(n)}$.
For each $1\leq i <j\leq a(n)$, the edge set $E_M(y)$ has the edge $(i,j)$, if either
\begin{enumerate}[label=(\roman*)]
	\item $j=i+1$, or
	\item there is some tape, such that, during the computation on $y$, $M$ visits the same block on that tape in both time-segments $i$ and $j$, and never visits that block during any time-segment strictly between $i$ and $j$.
\end{enumerate}
\end{definition}

In a directed acyclic graph $G$, we say that a vertex $u$ is a predecessor of a vertex $v$, if there exists a directed path from $u$ to $v$.

\begin{lemma}\label{lemma:ppst_graph_result} \cite{PPST83}
	For every $y$, the graph $G_M(y)$ satisfies the following:
	\begin{enumerate}
		\item Each vertex in $G_M(y)$ has degree $O(1)$.
		\item There exists a set of vertices $J\subset V_M(y)$ in $G_M(y)$, of size $O\brac{\frac{a(n)}{\log^*{a(n)}}}$ such that every vertex of $G_M(y)$ has at most $O\brac{\frac{a(n)}{\log^*{a(n)}}}$ many predecessors in the induced subgraph on the vertex set $V_M(y)\setminus J$.
	\end{enumerate}
	We note that the constants here might depend on the number of tapes of $M$.
\end{lemma}

\begin{lemma}\label{lemma:out_small_no_inp_blocks}
	Let $\epsilon>0$ be any constant and $\mc Y\subseteq \set{0,1}^m$ be any subset of the inputs. For all (large enough) $n$, there exists a subset $Y\subseteq \mc Y$ of size $\norm{Y} \geq \norm{\mc Y}\cdot 2^{-\epsilon nk}$, and subsets $S_1,\dots,S_k\subseteq [k]$ of size $\epsilon k$, such that for each $y = (f,x)\in Y$, and each $i\in[k]$, the $i$\textsuperscript{th} block (of length $n$) of $f^{\otimes n}(x)$ can be written as a function of $x|_{S_i}$ and the truth-table of $f$.
\end{lemma}
\begin{proof}
	For input $y = (f,x)\in \set{0,1}^m$, let $J(y) \subset V_M(y)$ be a set as in Lemma \ref{lemma:ppst_graph_result}.

Let $C(y)$ denote the following information about the computation of $M$:
\begin{enumerate}[label=(\roman*)]
	\item The internal state of $M$ at the end of each time-segment.
	\item The position of all tape heads at the end of each time-segment.
	\item For each time segment in $J(y)$, and for each tape of $M$, the final transcription (of length $n$) of the block that was visited on this tape during this segment.
\end{enumerate}

Let $g:\mc Y\to\set{0,1}^*$ be the function given by $g(y) = (G_M(y),J(y),C(y))$.
Observe that the output of $g$ can be described using $O\brac{k\log_2{k} +\frac{nk}{\log^*{k}}}$ bits.
By our assumption that $k(n) = \omega(1)$ and $k(n) = o(\log_2 n)$, we have that for large $n$, this is at most $\epsilon nk$ bits.
Hence, there exists a set $Y \subseteq \mc Y$ of size $\norm{Y}\geq \norm{\mc Y}\cdot 2^{-\epsilon nk}$, such that for each $y\in Y$, $g(y)$ takes on some fixed value $(G = (V,E),J,C)$.

Now, consider any $y = (f,x)\in Y$.
The machine writes the $k$ blocks of the output $f^{\otimes n}(x)$ on the output tape in the last $k$ time segments before halting.
For each of these time segments, the corresponding vertex in $G$ has at most $O\brac{\frac{k}{\log^*{k}}} \leq \epsilon k$ predecessors in the induced subgraph on $V\setminus J$.
These further correspond to at most $\epsilon k$ distinct blocks of $y$ that are visited (on the input tape) during these predecessor time segments.
Since the relevant block transcriptions at the end of time segments for vertices in $J$ are fixed in $C$, each output block can be written as a function of at most $\epsilon k$ blocks of $y$.
For the $i$\textsuperscript{th} block of output, without loss of generality, this includes the first block of $y$, which contains the truth table of $f$, and blocks of $x$ which indexed by some subset $S_i\subseteq [k]$ of size $\epsilon k$.

\end{proof}

\begin{proof}[Proof of Theorem \ref{thm:tm_advice_lb}]
Let $\delta=\frac{1}{8}$.
Fix some sufficiently large $n$, and a $(\delta k,\delta k)$-rigid function $f_0:\set{0,1}^k\to\set{0,1}^k$. The existence of such a function is guaranteed by Proposition \ref{prop:rig_func_exist}.
By Conjecture \ref{conj:rig_amp_to_block_rig}, the function $f_0^{\otimes n}$ is an $(\epsilon nk, \epsilon k)$-block-rigid function for some constant $\epsilon >0$.
On the other hand, Lemma $\ref{lemma:out_small_no_inp_blocks}$, with $\mc Y = \set{(f_0,x): x\in\set{0,1}^{nk}}$, shows that $f_0^{\otimes n}$ is not an $(\epsilon nk, \epsilon k)$-block-rigid function for any constant $\epsilon > 0$.
\end{proof}

We now restate and prove Theorem \ref{intro_thm:tm_lb_explicit}.

\begin{theorem}\label{thm:tm_lb_explicit}
Let $t:\N\to\N$ be any function such that $t(n) = \omega(n)$.
Assuming Conjecture $\ref{conj:rig_amp_to_block_rig}$, there exists a function $f:\set{0,1}^*\to \set{0,1}^*$ such that
\begin{enumerate}
	\item On inputs $x$ of length $n$ bits, the output $f(x)$ is of length at most $n$ bits.
	\item The function $f$ is computable by a multi-tape deterministic Turing machine that runs in time $O(t(n))$ on inputs of length $n$.
	\item The function $f$ is not computable by any multi-tape deterministic Turing machine that takes advice and runs in time $O(n)$ on inputs of length $n$.
\end{enumerate}
\end{theorem}

\begin{proof}
	Let $k:\N\to\N$ be a function such that $k(n) = \omega(1)$, $k(n) = o(\log_2{n})$, and $nk2^k \leq t(2^kk+nk)$.
	The theorem then follows from Observation \ref{obs:sup_lin_tm_ub} and Theorem $\ref{thm:tm_advice_lb}$.
\end{proof}

\begin{remark}
\begin{enumerate}[label=(\roman*)]
	\item We note that for the proof of Theorem \ref{thm:tm_advice_lb}, it suffices to find, for infinitely many $n$, a single function $f:\set{0,1}^{k(n)}\to\set{0,1}^{k(n)}$ such that $f^{\otimes n}$ is an $(\epsilon nk, \epsilon k)$-block-rigid function, where $\epsilon >0$ is a constant. This would show that $M$ cannot give the correct answer to $\probname$ for inputs of the form $(f,x)$, where $x\in \set{0,1}^{nk}$.
	\item For the proof of Theorem \ref{thm:tm_advice_lb}, it can be shown that it suffices for the following condition to hold for infinitely many $n$, and some constant $\epsilon >0$. Let $S_1,\dots,S_k\subseteq [k]$ be fixed sets of size $\epsilon k$, and $f:\set{0,1}^k\to\set{0,1}^k$ be a function chosen uniformly at random. Then, with probability at least $1-2^{-\omega(k\log_2{k})}$, the function $f^{\otimes n}$ is  an $(\epsilon nk, \epsilon k)$-block-rigid function against the fixed sets $S_1,\dots,S_k$. We note that the probability here is not good enough to be able to union bound over $S_1,\dots,S_k$ and get a single function as mentioned in the previous remark.
\end{enumerate}
\end{remark}

Essentially the same argument as that of Theorem \ref{thm:tm_advice_lb} also proves Theorem \ref{intro_thm:tm_lb_blk_rig}, which we restate below.
\begin{theorem}\label{thm:tm_lb_blk_rig}
	Let $k: \N \to \N$ be a function such that $k(n)=\omega(1)$ and $k(n)=2^{o(n)}$, and $f_n:\set{0,1}^{nk(n)}\to\set{0,1}^{nk(n)}$ be a family of $(\epsilon nk(n), \epsilon k(n))$-block-rigid functions, for some constant $\epsilon >0$.
	Let $M$ be any multi-tape deterministic linear-time Turing machine that takes advice.
	Then, there exists $n\in \N$, and $x\in \set{0,1}^{nk(n)}$, such that $M(x)\not= f_n(x)$.
\end{theorem}
The above theorem makes it interesting to find families of block-rigid functions that are computable in polynomial time.

\section{Size-Depth Tradeoffs}\label{sec:ckt_lb}

In this section, we will consider boolean circuits over a set $F$.
These are directed acyclic graphs with each node $v$ labelled either as an input node or by an arbitrary function $g_v:F\times F\to F$.
The input nodes have in-degree 0 and all other nodes have in-degree 2.
Some nodes are further labelled as output nodes, and they compute the outputs (in the usual manner), when the inputs are from the set $F$.
The size of the circuit is defined to be the number of edges in the graph.
The depth of the circuit is defined to be the length of a longest directed path from an input node to an output node.

Valiant \cite{Val77} showed that if $A\in \F^{n\times n}$ is an $(\epsilon n, n^{\epsilon})$-rigid matrix for some constant $\epsilon >0$, then the corresponding function cannot be computed by an $O(n)$-size and $O(\log_2{n})$-depth linear circuit over $\F$.
By a linear circuit, we mean that each gate computes a linear function (over $\F$) of its inputs.
A similar argument can be used to prove the following.

\begin{lemma} \cite{Val77} \label{lemma:rgd_fn_ckt_lb}
	Suppose $f:\set{0,1}^{nk}\to \set{0,1}^{nk}$ is an $(\epsilon nk, k^{\epsilon})$-block-rigid function, for some constant $\epsilon >0$. Then, the function $g:(\set{0,1}^{n})^k\to (\set{0,1}^{n})^k$ corresponding to $f$ cannot be computed by an $O(k)$-size and $O(\log_2{k})$-depth circuit over the set $F = \set{0,1}^n$.
\end{lemma}

%

\begin{theorem}\label{thm:ckt_lb_explicit}
	Let $k:\N\to\N$ be a function such that $k(n) = \omega(1)$ and $k(n) =  o(\log_2{n})$.
	Let $m = m(n):= 2^kk + nk$, where $k = k(n)$.
	
	Assuming Conjecture \ref{conj:rig_amp_to_block_rig}, the problem $\probname$ is not solvable by $O(k)$-size and $O(\log_2{k})$-depth circuits over the set $F = \set{0,1}^n$.
	Here, the input $(f,x)$ to the circuit is given in the form of $k+1$ elements in $\set{0,1}^n$, the first one being the truth table of $f$, and the remaining $k$ being the blocks of $x$.
\end{theorem}
\begin{proof}
	Let $\delta = \frac{1}{8}$. Fix some large $n$, and a $(\delta k,\delta k)$-rigid function $f_0:\set{0,1}^k\to\set{0,1}^k$, where $k = k(n)$.
	The existence of such a function is guaranteed by Proposition \ref{prop:rig_func_exist}.
	Assuming Conjecture \ref{conj:rig_amp_to_block_rig}, the function $f_0^{\otimes n}$ is an $(\epsilon nk, \epsilon k)$-block-rigid function, for some universal constant $\epsilon >0$.
	By Lemma \ref{lemma:rgd_fn_ckt_lb}, the corresponding function on $(\set{0,1}^n)^k$ cannot be computed by an $O(k)$-size and $O(\log_2{k})$-depth circuit over $F = \set{0,1}^n$.
	Since $f_0$ can be hard-wired in any circuit solving $\probname$, we have the desired result.
\end{proof}

%


\section{Rigid Matrices and Rigid Functions}\label{sec:rig_mat_vs_func}

A natural question to ask is whether the functions corresponding to rigid matrices are rigid functions or not.

\begin{conjecture}\label{conj:rig_mat_are_rig_func}
	There exists a universal constant $c>0$ such that whenever $A\in  \F_2^{n\times n}$ is an $(r,s)$-rigid matrix, the corresponding function $A: \F_2^n\to  \F_2^n$ is a $(cr,cs)$-rigid function.
\end{conjecture}

We show that a positive answer to the above resolves a closely related conjecture by Jukna and Schnitger \cite{JS11}.

\begin{definition}
Consider a depth-2 circuit, with $x=(x_1,\dots, x_n)$ as the input variables, $w$ gates in the middle layer, computing boolean functions $h_1,\dots,h_w$ and $m$ output gates, computing boolean functions $g_1,\dots,g_m$.
The circuit computes a function $f=(f_1,\dots,f_m): \F_2^n\to \F_2^m$ satisfying $f_i(x_1\dots,x_n) = g_i(x,h_1(x),\dots,h_w(x))$, for each $i\in[m]$.
The width of the circuit is defined to be $w$.
The degree of the circuit is defined to be the maximum over all gates $g_i$, of the number of wires going directly from the inputs $x_1,\dots,x_n$ to $g_i$.
\end{definition}

We remark that Lemma \ref{lemma:out_small_no_inp_blocks} essentially shows that any function computable by a deterministic linear-time Turing Machine has a depth-2 circuit of small width and small `block-degree'.

\begin{conjecture}\cite{JS11}\label{conj:lin_depth_two_ckt}
Suppose $f: \F_2^n\to \F_2^n$ is a linear function computable by a depth-2 circuit with width $w$ and degree $d$.
Then, $f$ is computable by a depth-2 circuit, with width $O(w)$, and degree $O(d)$, each of whose gates compute a linear function.	
\end{conjecture}

\begin{observation}
Conjecture \ref{conj:rig_mat_are_rig_func} $\implies$ Conjecture \ref{conj:lin_depth_two_ckt}.
\end{observation}
\begin{proof}
	Suppose $f: \F_2^n\to \F_2^n$ is a linear function computable by a depth-2 circuit with width $w$ and degree $d$.
	Then, there exists a set $X\subseteq  \F_2^n$ of size at least $2^{n-w}$, such that for each $x\in X$, the value of the functions computed by the gates in the middle layer is the same.
	Hence, for each $x\in X$, each element of $f(x)$ can be written as a function of at most $d$ elements of $x$.
	This shows that $f$ is not an $(w,d)$-rigid function.
	Assuming Conjecture \ref{conj:rig_mat_are_rig_func}, the matrix $A\in  \F_2^{n\times n}$ for the function $f$ is not a $(cw,cd)$-rigid matrix, for some constant $c>0$.
	Then, $A = B+C$, where the rank of $B$ is at most $cw$, and $C$ has at most $cd$ non-zero entries in each row.
	Now, there exist matrices $B_1\in \F_2^{n\times cw}$ and $B_2\in \F_2^{cw\times n}$, such that $B=B_1B_2$.
	Then, $f$ is computable by a linear depth-2 circuit with width $cw$ and degree $cd$, where the middle layer computes output of the function corresponding to $B_2$.
\end{proof}


\section{Future Directions}\label{sec:matrix_problems}

In this section, we state some well known problems related to matrices. It seems interesting to study the block-rigidity of the these functions. 

\subsection{Matrix Transpose}

The matrix-transpose problem is described as follows: 

\begin{itemize}
	\item[-] \emph{Input:} A matrix $X\in \F_2^{n\times n}$ as a vector of length $n^2$ bits, in row-major order, for some $n\in\N$.
	\item[-] \emph{Output:} The matrix $X$ column-major order (or equivalently, the transpose of $X$ in row-major order).
\end{itemize}

It is well known (see \cite{DMS91} for a short proof) that the above problem can be solved on a 2-tape Turing machine in time $O(N\log{N})$, on inputs of length $N=n^2$.
We believe that this cannot be solved by Turing machines in linear-time, and that the notion of block-rigidity might be a viable approach to prove this.
Next, we observe some structural details about the problem.

The matrix-transpose problem computes a linear function, whose $N\times N$ matrix $A$ on inputs of length $N=n^2$ is described as follows. 
For each $i,j\in[n]$, let $e_{ij}\in\F_2^{n\times n}$ denote the matrix whose  $(i,j)$\textsuperscript{th} entry is 1 and rest of the entries are zero. 
The matrix $A$ is an $N\times N$ matrix made up of $n^2$ blocks, with the $(i,j)$\textsuperscript{th} block equal to $e_{ji}$. 

Using a similar argument as in Observation \ref{obs:block_rig_func_game_val}, one can show that the value of the following game captures the block-rigidity of the matrix-transpose function. 
Fix integers $n\in \N$, $1\leq s<n$, and a collection $\mc S = (S_1,\dots,S_n)$, where each $S_i \subseteq[n]$ is of size $s$.
We define an $n$-player game $\mc G_{\mc S}$ as follows:
A verifier chooses a matrix $X\in \F_2^{n\times n}$, with each entry being chosen uniformly and independently.
For each $j\in[n]$, player $j$ is given the rows of the matrix indexed by $S_j$, and they answer $y_j  \in  \F_2^n$.
The verifier accepts if and only if for each $j\in[n]$, $y_j$ equals the $j$\textsuperscript{th} column of $X$.

\begin{conjecture}
	There exists a constant $c>0$ such that the function given by the matrix $A$ is a $(c n^2, c n)$-block-rigid function. Equivalently, for each collection $\mc S$ with set sizes as $cn$, the value of the game $\mc G_{\mc S}$ is at most $2^{-cn^2}$. 
\end{conjecture}

We note that the above game is of independent interest from a combinatorial point of view as well. 
Basically, it asks whether there exists a large family of $n\times n$ matrices, in which each column can be represented as some function of a small fraction of the rows.
The problem of whether the matrix $A$ is a block-rigid matrix is also interesting. 
This corresponds to the players in the above game using strategies which are linear functions.

\subsection{Matrix Product}

The matrix-product problem is described as follows: 

\begin{itemize}
	\item[-] \emph{Input:} Matrices $X, Y\in \F_2^{n\times n}$ as vectors of length $n^2$ bits, in row-major order, for some $n\in\N$.
	\item[-] \emph{Output:} The matrix $Z=XY$ in row-major order.
\end{itemize}

The block-rigidity of the matrix-product function is captured by the following game:
Fix integers $n\in \N$, $1\leq s<n$, and collections $\mc S = (S_1,\dots,S_n)$, $\mc T = (T_1,\dots,T_n)$ where each $S_i,T_i \subseteq[n]$ is of size $s$.
We define a $n$-player game $\mc G_{\mc S, \mc T}$ as follows:
A verifier chooses matrices $X,Y\in \F_2^{n\times n}$, with each entry being chosen uniformly and independently.
For each $j\in[n]$, player $j$ is given the rows of the matrices $X$ and $Y$ indexed by $S_j$ and $T_j$ respectively, and they answer $y_j  \in  \F_2^n$.
The verifier accepts if and only if for each $j\in[n]$, $y_j$ equals the $j$\textsuperscript{th} row of $XY$.

\begin{conjecture}
	There exists a constant $c>0$ such that for each $\mc S, \mc T$ with set sizes as $cn$, the value of the game $\mc G_{\mc S, \mc T}$ is at most $2^{-cn^2}$. 
\end{conjecture}

One may change the row-major order for some (or all) of the matrices to column-major order. It is easy to modify the above game in such a case.

\bibliographystyle{alpha}
\bibliography{ref}

\end{document}